\documentclass[dvipsnames]{article}

\usepackage[a4paper, total={151mm,247mm}]{geometry}
\usepackage{amsfonts}       %
\usepackage{amsmath}
\usepackage{amsthm} %
\usepackage{bm} %
\usepackage{graphicx}
\usepackage[all,cmtip]{xy}%
\usepackage{caption,subcaption} %

\usepackage{authblk} %

\usepackage{booktabs}

 \usepackage{multirow} %

\usepackage[shortlabels]{enumitem} %

\DeclareMathOperator{\EX}{\mathbb{E}}%

\usepackage[textsize=tiny]{todonotes}

\newtheorem{theorem}{Theorem}

\newtheorem{remark}{Remark}
\newtheorem{definition}{Definition}

\usepackage[acronym,nowarn]{glossaries}
\makeglossaries
\glsdisablehyper

\usepackage{tikz}
\tikzset{
  invisible/.style={opacity=0},
  visible on/.style={alt={#1{}{invisible}}},
  alt/.code args={<#1>#2#3}{%
    \alt<#1>{\pgfkeysalso{#2}}{\pgfkeysalso{#3}} %
  },
}
\usetikzlibrary{arrows,positioning,calc}

\usepackage[sorting=none]{biblatex}

\author[1]{Wouter A.C. van Amsterdam}
\affil[1]{Julius Center for Health Sciences and Primary Care, University Medical Center Utrecht, Utrecht University, The Netherlands, Department of Data Science and Biostatistics}

\title{A causal viewpoint on prediction model performance under changes in case-mix: discrimination and calibration respond differently for prognosis and diagnosis predictions}

\begin{document}

\maketitle

\begin{abstract}

Prediction models need reliable predictive performance as they inform clinical decisions, aiding in diagnosis, prognosis, and treatment planning. The predictive performance of these models is typically assessed through discrimination and calibration.
Changes in the distribution of the data impact model performance and there may be important changes between a model's current application and when and where its performance was last evaluated.
In health-care, a typical change is a shift in \emph{case-mix}.
For example, for cardiovascular risk management, a general practitioner sees a different mix of patients than a specialist in a tertiary hospital.

This work introduces a novel framework that differentiates the effects of case-mix shifts on discrimination and calibration based on the causal direction of the prediction task. When prediction is in the causal direction (often the case for prognosis preditions), calibration remains stable under case-mix shifts, while discrimination does not. Conversely, when predicting in the anti-causal direction (often with diagnosis predictions), discrimination remains stable, but calibration does not.

A simulation study and empirical validation using cardiovascular disease prediction models demonstrate the implications of this framework. The causal case-mix framework provides insights for developing, evaluating and deploying prediction models across different clinical settings, emphasizing the importance of understanding the causal structure of the prediction task.

\end{abstract}

keywords: calibration, discrimination, case-mix, causal inference, prediction model, external validation

\section{Introduction}

Clinicians use prediction models for medical decisions, for example when making a diagnosis, estimating a patient's prognosis, or when making triaging or treatment decisions. %
When basing a medical decision on a prediction model it is important to know how reliable the model's predictions are, i.e. what is the model's \emph{predictive performance}, typically measured with \textit{discrimination} and \textit{calibration} in the case of binary outcomes.
Discrimination measures how well a prediction model separates \emph{positive} cases from \emph{negative} cases, whereas \emph{calibration} measures how well predicted probabilities align with observed event rates. %

An issue with predictive performance is that there may be important changes between when a model's predictive performance was last evaluated, and when and where it is used, meaning that the underlying \emph{data distribution} may have changed.
No model can have good predictive performance under all arbitrary changes in the data distribution,
but we may consider one important class of changes in distribution described with the term `case-mix'.
For instance, when comparing cardiovascular risk management in the general practitioner setting with a tertiary hospital setting, the frequency of certain comorbidities and risk factors will be different across settings.
Typically the tertiary center will encounter more high risk patients, so their `case-mix' is different than in the general practitioner setting.
Another change in case-mix is the frequency of myocardial infarction in patients presenting with chest pain at either the general practitioner or in those referred to acute cardiac care centers.
We present a formal definition of a shift in case-mix using the language of causality.
The \emph{independent causal mechanisms} principle states that when viewing the joint distribution of observed data trough a mechanistic causal lens, where \emph{effects} are created from inputs (\emph{causes}) by a causal \emph{mechanism}, the distribution of the inputs (causes) is independent of the mechanism that produces the outputs from the inputs \cite{petersElementsCausalInference2017}.
A natural causal definition of a shift in case-mix is thus a shift in the marginal distribution of the \emph{causes}, and we may make the additional assumption that the \emph{mechanism} is the same across environments.
When applied to the above example this would mean that though tertiary case center patients have a different distribution of risk factors than patients from the general practitioner setting, one may hope that given knowledge of sufficient risk-factors, the risk of cardiovascular disease for two patients with the same values of risk-factors is the same regardless of what setting they are in.

This causal definition of a shift in case-mix implies that when the prediction task is in the \emph{causal} direction versus in the \emph{anti-causal} direction, a change in case-mix has a different interpretation.
Inferring a diagnosis is typically prediction in the \emph{anti-causal} direction, meaning predicting the cause (=underlying diagnosis) based on its effects (=symptoms), and here a change in case-mix is a change in distribution of the prediction target (the diagnosis).
In contrast, prognosis is typically prediction in the \emph{causal} direction, meaning a future outcome predicted from current patient characteristics, and here a shift in case-mix is a change in the distribution of patient characteristics.
Importantly, depending on the prediction direction, either calibration \emph{or} discrimination is preserved under shifts in case-mix, but not both.
The crucial insight underlying our results is that a prediction model's \emph{discrimination} depends on the distribution of the features given the outcome ($X$ given $Y$) and is thereby invariant to changes in the marginal distribution of the outcome.
Conversely, \emph{calibration} depends on the distribution of the outcome given the features ($Y$ given $X$) and is thereby invariant to changes in the marginal distribution of the features.
See Figure \ref{fig:main} for a schematic overview.

Our result shows that the causal direction of the prediction has important implications for the development, evaluation and deployment of prediction models.
For example, when evaluating a model used for prognosis across different settings, changes in discrimination are expected under shifts in case-mix, but changes in calibration are not, and vice-versa for diagnostic models.
When re-evaluating a prognostic model in a different setting, a change in discrimination is expected and thus no cause for concern.
However, a marked change in calibration may warrant further research.
Another perhaps unexpected result is that when a model is evaluated across different environments, the observation that \emph{either} discriminaton or calibration remains stable is a stronger sign of robustness to changes in environment than when both remain stable. The reason is that when both remain stable, this proves that the testing environments were not meaningfully different. Only when either discrimination or calibration changes and the other is stable, we gain some confidence that the model remains robust across different environments.
This perspective helps developers and guideline makers judge where and when a prediction model has dependable predictive performance.
Also, depending on the task and whether discrimination or calibration is more important, prediction model developers may improve the robustness of their model to changes in case-mix by only including variables in the prediction model that are either all causal or all anti-causal but not mixed, when possible.

 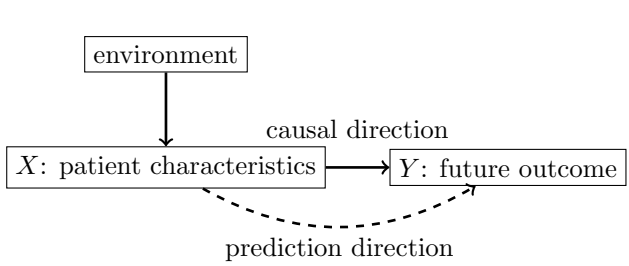
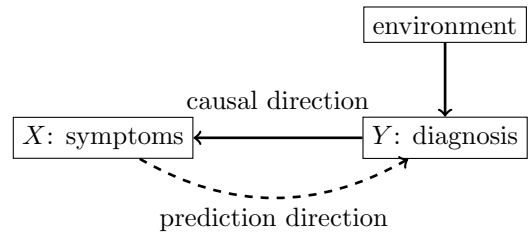
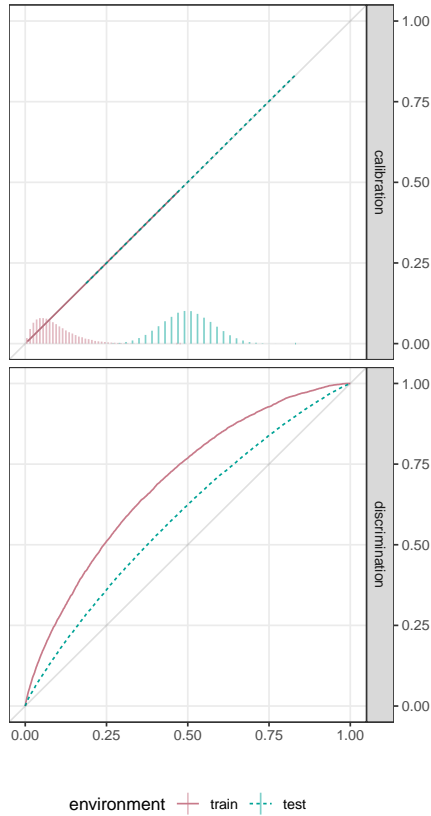
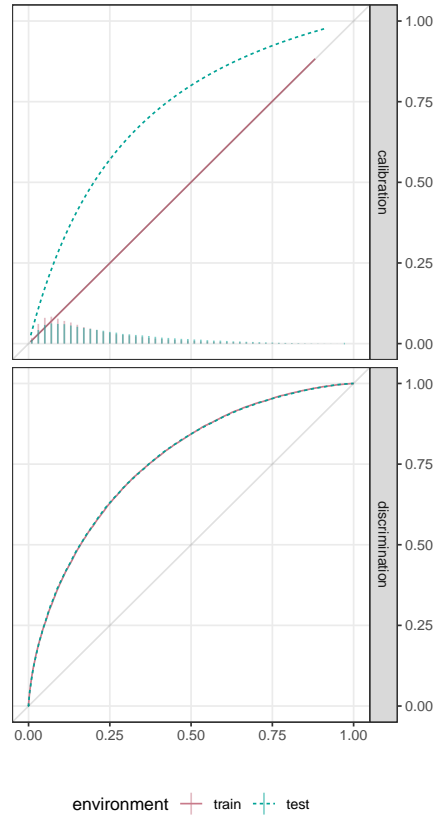
\begin{figure}[htpb]
	\begin{subfigure}[b]{0.4\textwidth}
     \centering
	     \begin{tikzpicture}

		 \node[draw, text centered] (x) {$X$: patient characteristics};
		 \node[draw, right of = x, node distance=4.5cm, text centered] (y) {$Y$: future outcome};

		 \node[draw, rectangle, above of = x, node distance=1.5cm, text centered] (e) {environment};

		 \draw[->, line width = 1] (x) --  (y) node[midway,above=.25cm] {causal direction};
         \draw[->, line width = 1] (e) -- (x);
		 \draw[->, line width = 1, dashed] (x) to[bend right] node[below] {prediction direction}  (y);

	     \end{tikzpicture}
         \caption{DAG for prediction models predicting in the \emph{causal} direction, as in many \emph{prognosis} settings (e.g. predict future heart attacks based on current age and cholesterol levels)}
    \label{fig:dag-causal}
     \end{subfigure}
\hfill
\hfill
	\begin{subfigure}[b]{0.4\textwidth}
     \centering
	     \begin{tikzpicture}

		 \node[draw, text centered] (x) {$X$: symptoms};
		 \node[draw, right of = x, node distance=4.5cm, text centered] (y) {$Y$: diagnosis};

		 \node[draw, rectangle, above of = y, node distance=1.5cm, text centered] (e) {environment};

		 \draw[->, line width = 1] (y) -- (x) node[midway,above=.25cm] {causal direction};
		 \draw[->, line width = 1] (e) -- (y);
		 \draw[->, line width = 1, dashed] (x) to[bend right] node[below] {prediction direction}  (y);
	     \end{tikzpicture}
         \caption{
        DAG for prediction models predicting in the \emph{anti-causal} direction as in many \emph{diagnosis} settings (e.g. predict presence of a current heart attack based on the presence of chest pain and electrocardiography abnormalities).
}
    \label{fig:dag-anticausal}
     \end{subfigure}
\vfill
\hfill
\hfill
\begin{subfigure}[c]{0.4\textwidth}
\centering
    \includegraphics[width=\textwidth]{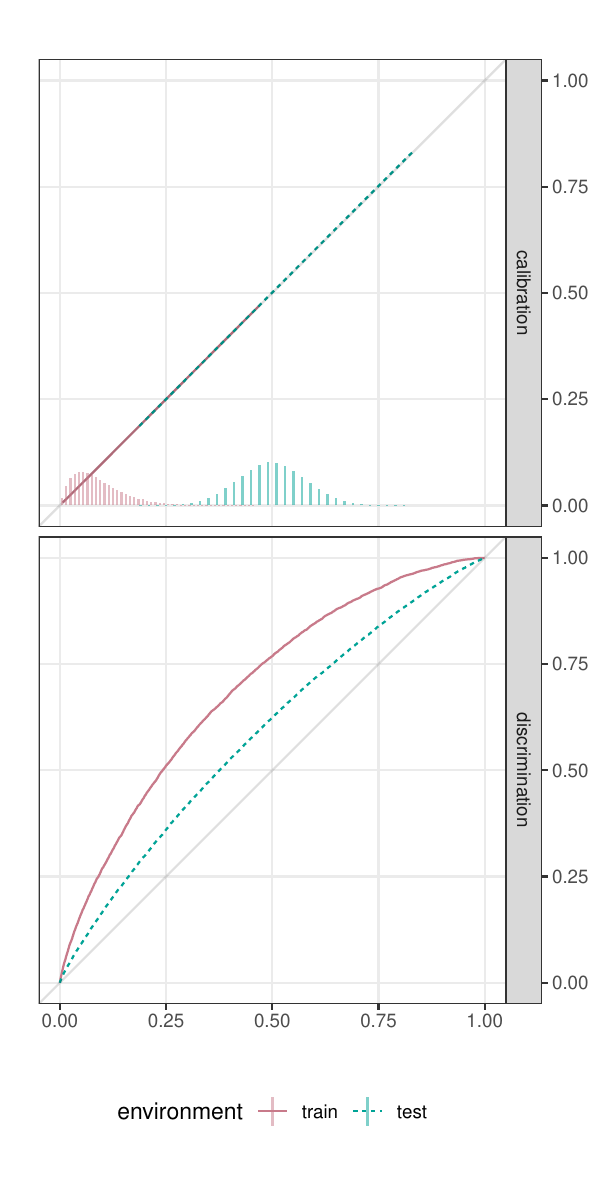}
    \caption{Between the training data and testing data, the \emph{calibration} remains the same (upper facet), but the \emph{discrimination} changes (lower facet)}
\label{fig:1causal}
\end{subfigure}
\hfill
\hfill
\hfill
\begin{subfigure}[c]{0.4\textwidth}
\centering
    \includegraphics[width=\textwidth]{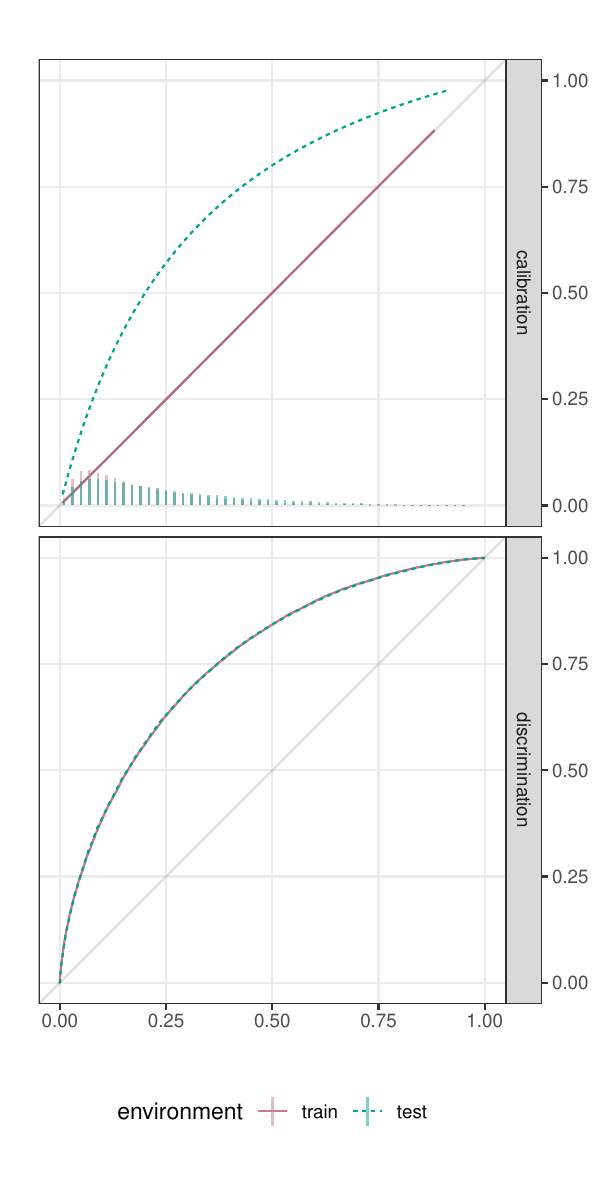}
    \caption{Between the training data and testing data, the \emph{calibration} changes (upper facet), but the \emph{discrimination} remains the same (lower facet)}
\label{fig:1anticausal}
\end{subfigure}
\caption{Overview of main results.
    Depending on the causal direction of the prediction, a shift in `case-mix' may be defined as either a shift in the marginal distribution of the \emph{features} $X$ for \emph{causal} prediction (\ref{fig:dag-causal}) or a shift in the marginal distribution of the \emph{outcome} $Y$ for \emph{anti-causal} prediction (\ref{fig:dag-anticausal}).
    With these definitions, for models predicting in the \emph{causal direction}, the \emph{calibration} will remain constant under case-mix shifts between the training data and the testing data but not the \emph{discrimination} (\ref{fig:1causal}). For models predicting in the \emph{anti-causal direction} the reverse is true (\ref{fig:1anticausal}).
    The calibration facets are calibration curves with on the horizontal axis the predicted probability and on the vertical axis the actual probability.
    The discrimination facets are receiver-operating-curves with on the horizontal axis 1 minus specificity and on the vertical axis sensitivity.
DAG: directed acyclic graph
}
\label{fig:main}
\end{figure}

To introduce the framework, we first review the concepts of discrimination and calibration and then we define changes in case-mix from a causal viewpoint.
Next put the two pieces together in a new framework and answer: when to expect what changes in predictive performance?
We illustrate the result with a simulation study and test the framework empirically in a systematic review of 1382 prediction models, where we find that prognostic models indeed have more variance in discrimination when tested in external validation studies.
Finally we discuss how this theory can be used in practice.

\pagebreak

\section{Notation and review of predictive performance: discrimination and calibration}

We consider prediction models of a binary \emph{outcome} $Y$ using \emph{features} $X$ with a prediction model $f: \mathcal{X} \to [0,1]$.
The features can come from an arbitrary (multi-)dimensional distribution.%
We will denote environments with an environment variable $E$ where for example $E=0$ may be a general practitioner setting, $E=1$ a community hospital and $E=2$ a university medical center \cite{bareinboimCausalInferenceDatafusion2016}.
With $P(.)$ we will denote (conditional) distributions or densities over random variables, for example $P(Y|X)$ denotes the distribution of outcome $Y$ given features $X$.%

\subsection{Discrimination: sensitivity, specificity and AUC}
The typical metrics of discrimination are sensitivity (sometimes called recall), specificity and AUC.
Sensitivity is the ratio of true positives over the total number of positive cases.
Specificity is the ratio of true negatives over the total number of negative cases.
To calculate sensitivity and specificity, we need to choose a threshold $0\leq \tau \leq 1$ for the output of $f(X)$ and label all $f(X) > \tau$ as \emph{positive} predicted cases and $f(X) \leq \tau$ as \emph{negative} predicted cases.
This results in a 2 by 2 table with predicted versus actual outcomes (sometimes called the `confusion table'), see Table \ref{tab:confusion}.
By varying $\tau$ between 0 and 1 we get a range of values for sensitivity and specificity.
Plotting these in the receiver-operating-curve and calculating the area under this curve we get the popular AUC metric or c-statistic.
Note that for calculating sensitivity we only need the \emph{positive} cases ($Y=1$), and for specificity we only need the \emph{negative} cases ($Y=0$).
\emph{Measures of discrimination depend on the distribution of the prediction (and thus the features) given the outcome.}
This immediately implies that if we were to only change the ratio of positive and negative cases through some hypothetical intervention, the sensitivity and specificity will remain unchanged, and thus the resulting AUC.
Therefore it is sometimes said that sensitivity and specificity are prevalence independent.

\begin{table}[]
	\centering
\begin{tabular}{llll}
                                                 &                        & \multicolumn{2}{c}{outcome ($Y$)}                                               \\ \cline{3-4} 
                                                 & \multicolumn{1}{l|}{}  & \multicolumn{1}{l|}{1}              & \multicolumn{1}{l|}{0}              \\ \cline{2-4} 
\multicolumn{1}{c|}{\multirow{2}{*}{prediction ($f(X) > \tau$)}} & \multicolumn{1}{l|}{1} & \multicolumn{1}{l|}{true positive}  & \multicolumn{1}{l|}{false positive} \\ \cline{2-4} 
\multicolumn{1}{c|}{}                            & \multicolumn{1}{l|}{0} & \multicolumn{1}{l|}{false negative} & \multicolumn{1}{l|}{true negative}  \\ \cline{2-4} 
                                                 &                        & sensitivity: $P(f(X)>\tau|Y=1)$  & specificity: $P(f(X) \le \tau | Y=0)$                        
\end{tabular}
\caption{Confusion table. By specifying a threshold $0 \leq \tau \leq 1$ for a prediction model $f: \mathcal{X} \to [0,1]$ and tabulating the results against the ground truth outcome $Y \in \{0,1\}$, we get the confusion table and can calculate metrics of discrimination such a sensitivity and specificity. }
\label{tab:confusion}
\end{table}

\subsection{Calibration}

Calibration measures how well predicted probabilities align with actual event rates.
In words, assume we take a particular value for the predicted probability of the outcome, say $\alpha=10\%$.
Then if we gather all cases for which $f(X) = \alpha$, then the model is calibrated for that value of $\alpha$ when the fraction of positive outcomes in this subset is exactly $\alpha$.
A prediction model is perfectly calibrated when this holds for all unique values that $f(X)$ attains.
For a formal definition, see Definition \ref{def:calibration} in the Appendix \ref{app:defs}.
Unfortunately, measuring discrimination with a single metric is much harder then measuring discrimination for practical \cite{vancalsterCalibrationAchillesHeel2019,vancalsterCalibrationHierarchyRisk2016} and theoretical reasons \cite{blasiokUnifyingTheoryDistance2023}, a problem we will ignore.
However, fundamentally, calibration measures the alignment between $f(X)$ and the probability of the outcome given $X$.
\emph{Measures of calibration are thus measures of the distribution of the outcome given the features ($Y$ given $X$).}

\pagebreak

\section{A causal framework for predictive performance under changes in case-mix}

Since discrimination depends on the distribution of the features given the outcome ($X$ given $Y$) but calibration on the distribution of the outcome given the features ($Y$ given $X$), we may expect metrics of discrimination and calibration to respond differently when changes occur in the marginal distribution of $X$ or $Y$.
In this section we first formalize the notion of a shift in \emph{case-mix} and how this depends on whether a prediction is in the \emph{causal} direction (future outcome given features) or the \emph{anti-causal} direction (e.g. disease given symptoms).
Then we will draw the connection between the two insights leading to our main result.

\subsection{A shift in case-mix is a change in the marginal distribution of the cause variable}
Inspired by the principle of independence of \emph{cause and mechanism} \cite{petersElementsCausalInference2017, schoelkopfCausalAnticausalLearning2012}, we define a shift in \emph{case-mix} between different environments (e.g. general practitioner versus hospital setting) as a change in the marginal distribution of the \emph{cause} variable.
When the prediction is in the \emph{causal} direction, a shift in case-mix is a change in the marginal distribution of the features $X$, whereas when the prediction is in the \emph{anti-causal} direction it is a shift in the marginal distribution of the outcome $Y$.

Finally, the prediction problem could be neither causal or anti-causal, but \emph{confounded} by another variable $Z$, in that case the shift is in the distribution of the confounder $Z$.
See Figure \ref{fig:dags} for directed acyclic graphs (DAGs) depicting these situations and Table \ref{tab:settings} for an overview with examples.
We give a formal definition in the Appendix \ref{def:casemix}.

\begin{table}[htpb]
	\centering
	\caption{different prediction settings}
	\label{tab:settings}

	\begin{tabular}{l|c|c|c}
		& anti-causal & causal & confounded \\ \hline
		shifted distribution & $Y$ & $X$ & $Z$  \\
		typical setting	& diagnosis & prognosis & prognosis  \\
        example outcome & pneumonia & survival & lung cancer diagnosis \\
		example features & temperature & age & yellow fingers \\
        Figure & \ref{fig:causal} & \ref{fig:anticausal} & \ref{fig:fork}
	\end{tabular}
\end{table}

 \begin{figure}[ht!]
	\begin{subfigure}[b]{0.3\textwidth}
     \centering
	     \begin{tikzpicture}

		 \node[text centered] (x) {$X$};
		 \node[right of = x, node distance=1.5cm, text centered] (y) {$Y$};

		 \node[draw, rectangle, above of = x, node distance=1.5cm, text centered] (e) {$E$};

		 \draw[->, line width = 1] (x) --  (y);
		 \draw[->, line width = 1] (e) -- (x);
	     \end{tikzpicture}
	\caption{causal prediction}
	\label{fig:causal}
     \end{subfigure}
\hfill
	\begin{subfigure}[b]{0.3\textwidth}
     \centering
	     \begin{tikzpicture}

		 \node[text centered] (x) {$X$};
		 \node[right of = x, node distance=1.5cm, text centered] (y) {$Y$};

		 \node[draw, rectangle, above of = x, node distance=1.5cm, text centered] (e) {$E$};

		 \draw[->, line width = 1] (y) --  (x);
		 \draw[->, line width = 1] (e) -- (y);
	     \end{tikzpicture}
	\caption{anti-causal prediction}
	\label{fig:anticausal}
     \end{subfigure}
\hfill
	\begin{subfigure}[b]{0.3\textwidth}
     \centering
	     \begin{tikzpicture}

		 \node[text centered] (x) {$X$};
		 \node[right of = x, node distance=1.5cm, text centered] (y) {$Y$};
		 \node[above of = y, node distance=1.5cm, text centered] (z) {$Z$};

		 \node[draw, rectangle, above of = x, node distance=1.5cm, text centered] (e) {$E$};

		 \draw[->, line width = 1] (z) --  (x);
		 \draw[->, line width = 1] (z) --  (y);
		 \draw[->, line width = 1] (e) -- (z);
	     \end{tikzpicture}
	\caption{confounded prediction}
	\label{fig:fork}
     \end{subfigure}
     \caption{directed acyclic graphs for 2-variable prediction problems with a shift in case-mix, meaning the environment variable only affects the marginal distribution of only the cause variable. The prediction is always made from feature $X$ to outcome $Y$, $E$ denotes the environment.}
     \label{fig:dags}
 \end{figure}
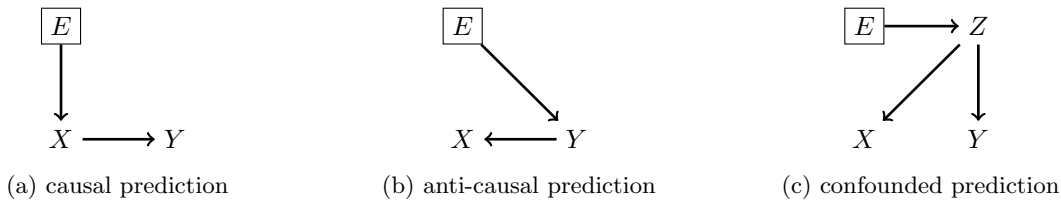

Each of the DAGs in Figure \ref{fig:dags} encodes different conditional idependencies.
Specifically the DAG in the causal direction (Figure \ref{fig:causal}) implies that $Y$ is independent of $E$ given $X$.
This entails that the distribution $P(Y|X)$ is \emph{transportable} across environments, so for different environments $E=0,1,\ldots$, $P(Y|X,E=0) = P(Y|X,E=1) = P(Y|X)$, but $P(X|Y)$ is not transportable: $P(X|Y,E) \neq P(X|Y)$.
Conversely, in the anti-causal direction (Figure \ref{fig:anticausal}) the distribution $P(X|Y)$ is transportable, meaning $P(X|Y,E) = P(X|Y)$, but not $P(Y|X)$.
In the confounded DAG (Figure \ref{fig:fork}) neither $P(Y|X)$ or $P(X|Y)$ is transportable.

In the DAGs in Figure \ref{fig:dags} the environment variable influences the cause variable ($X,Y$ or $Z$) but not the effect variable ($Y$ or $X$).
Why exclude arrows from the environment to the effect variable in the definition of a shift in case-mix?
First, when viewed as a mechanistic description of the data generating process, the principle of independence of cause and mechanism states that the distribution of the cause variable is independent of the mechanism that produces the effect variable \cite{petersElementsCausalInference2017}.
Also, if there is an arrow from environment to the effect variable, neither $P(Y|X)$ or $P(X|Y)$ are transportable across environments so nothing can be said regarding the calibration and discrimination of a prediction model on an unseen environment based on data from the observed environments only and the assumptions expressed in the DAG.
Finally, in clinical settings it may be reasonable based on temporal ordering and patient selection mechanisms to assume that \emph{at least} the distribution of the cause variable differs between environments, but maybe not the effect given the cause.
See the Appendix \ref{app:examples} for several concrete medical examples where these assumptions may hold.

\subsection{Main result: discrimination and calibration respond differently to changes in case-mix depending on the causal direction of the prediction}

With the DAGs describing the different possible shifts in case-mix under consideration and the definitions of discrimination and calibration we can now state our main result, of which the formal versions are presented in the Appendix \ref{app:theorems}.

\emph{When predicting in the anti-causal direction (often with diagnosis predictions), a shift in case-mix across environments means a shift in the marginal distribution of the outcome, and discrimination remains stable across environments but not calibration. Conversely, when predicting in the causal direction (often with prognosis predictions), a shift in case-mix across environments means a shift in marginal distribution of the features, and calibration remains stable, but not discrimination.}

For prediction in the causal direction, when $f$ is perfectly calibrated on an environment, it will remain perfectly calibrated under shifts of the marginal distribution of the features (see Theorem \ref{th:calib} in the Appendix).
Note that when $f$ is not perfectly calibrated and this mis-calibration depends on $X$, in general the average calibration will also change when predicting in the causal direction.
An important implication of this result is that \emph{discrimination or calibration may be preserved under changes in case-mix, but typically not both}\footnote{For predicting in the causal direction, examples where the AUC remains constant across environments may be constructed. Consider having a mixture of beta-distributions for $P(Y|X)$ with their modes $\mu_1 = 0.25$ and $\mu_2 = 0.75$. Shifting $\mu_1$ to 0 will increase AUC, shifting $\mu_2$ to 0.5 will decrease AUC. By shifting both modes at the same time, these changes can be set to balance out.}.

As a remark, we note that perfectly calibrated models obviously cannot be better calibrated in other environments, so any change in calibration necessarily implies a worsening of calibration.
For discrimination, this is not automatically the case. In fact, models show better discrimination in other environments then the training data when the distribution of outcome probabilities becomes less concentrated around 50\%.

\section{Simulation and empirical evaluation}

\subsection{Illustrative simulation}

Our main result has important implications when interpreting changes in predictive performance across environments.
To illustrate our result we now present a simulation study.
Consider two prediction models, one is a prognostic model predicting in the causal direction, the other a diagnostic model predicting in the anti-causal direction.
Denoting $\sigma^{-1}(p)=\log \frac{p}{1-p}$ as the logit function and $\mathcal{N}$ the Gaussian distribution,
the data-generating mechanisms are:

\begin{align*}
    \label{eq:dgm-prognosis}
    \text{prognosis:} &                     & \text{diagnosis:} & \\
    P_y &\sim \text{Beta}(\alpha_e,\beta_e) & y &\sim \text{Bernouli}(P_e) \\
    x   &= \sigma^{-1}(P_y)                 & x &\sim \mathcal{N}(y, 1) \\
    y   &\sim \text{Bernoulli}(P_y)         &   &
\end{align*}

We evaluate both models in three hypothetical environments: a screening environment with low outcome prevalence, a general practitioner setting with intermediate prevalence and a hospital setting with high prevalence.
For the prognosis model, the marginal distribution of $X$ depends on the environment through $\alpha_e,\beta_e$, but not the distribution of $Y$ given $X$.
For the diagnosis model, the marginal distribution of $Y$ depends on the environment through $P_e$, but not the distribution of $X$ given $Y$.
The different values for these parameters in the simulation are given in Table \ref{tab:simprms}.

\begin{table}[htp]
    \centering
\begin{tabular}{lllll}
          &                                &           &     &          \\
task      & \multicolumn{1}{l|}{parameter} & screening & general practitioner  & hospital \\ \hline
prognosis & \multicolumn{1}{l|}{$\alpha$}  & 2         & 5   & 10       \\
          & \multicolumn{1}{l|}{$\beta$}   & 20        & 10  & 20       \\
diagnosis & \multicolumn{1}{l|}{p}         & 0.2       & 1/3 & 0.5     
\end{tabular}
\caption{Values for different simulation parameters in three hypothetical environments.}
\label{tab:simprms}
\end{table}

In Figure \ref{fig:overview} we show the results of training a prediction model in the screening environment and evaluating it either in the same environment (`internal validation') or in a different environment (`external validation').
For the prognostic model the calibration remains the same across environments; %
the discrimination changes across environments.
For the diagnostic model, the reverse is true.

\begin{figure}[htpb]
    \centering
    \includegraphics[width=0.8\textwidth]{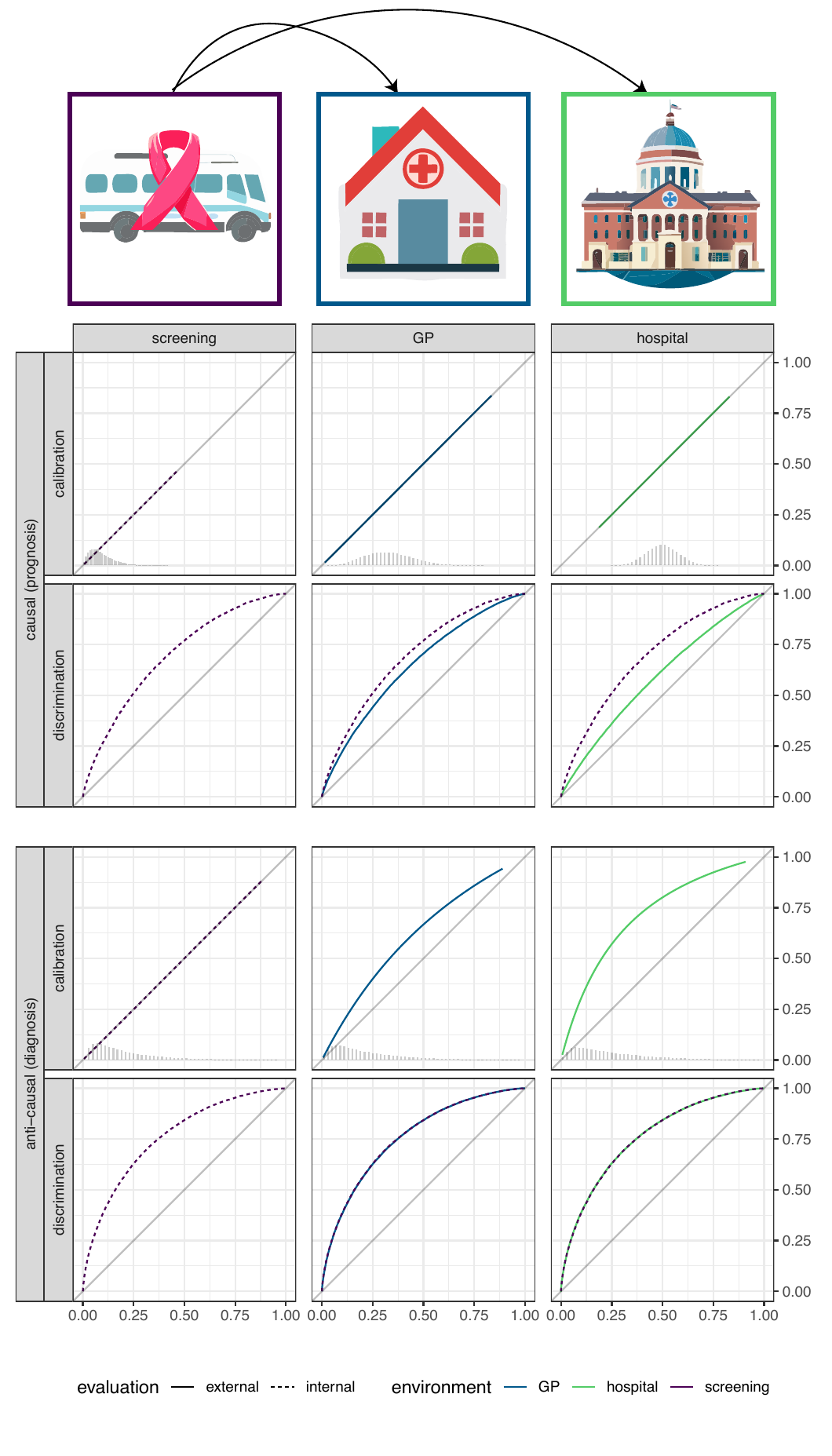}
    \caption{Overview figure of illustrative simulation experiment of a model trained on data from a screening environment, and evaluated on either the screening environment (`internal validation') or the general practitioner (GP) environment or the hospital environment (`external validation'), with increasing outcome probabilities. For models predicting in the \textit{anti-causal} direction (e.g. diagnostic models), a shift in case-mix entails a shift in the distribution of the outcome, so discrimination remains the same but calibration changes. For models predicting in the \textit{causal} direction (e.g. prognosis models), a shift in case-mix entails a shift in the distribution of the features, so calibration remains the same but the discrimination changes.
        The discrimination facets are receiver-operating-curves with on the horizontal axis 1 minus specificity and on the vertical axis sensitivity.
        The calibration facets are calibration curves with on the horizontal axis the predicted probability and on the vertical axis the actual probability.
    }
    \label{fig:overview}
\end{figure}

By repeating this process for each of the three environments, each time training on one environment and evaluating on all environments for both the causal prediction model and the anti-causal model, we get in total six models, each evaluated three times.
We measure discrimination with AUC and calibration error as the average absolute difference between the predicted outcome probability and the actual outcome probability for each observation:  $\frac{1}{N} \sum_i^N \left| P(Y=1|X=x_i) - f(x_i)\right|$ (analogous to the Integrated Calibration Error defined in \cite{austinIntegratedCalibrationIndex2019}).
Plotting these 18 points on 6 lines in 2 dimensions leads to an interesting pattern, where the models predicting in the causal direction are easily discernible from those predicting in the anti-causal direction (Figure \ref{fig:combined}).
In the Appendix \ref{app:sims} we provide visualizations of $P(Y|X)$ and $P(X|Y)$ for the different environments and tasks.

\begin{figure}[htpb]
    \centering
    \includegraphics[width=0.8\textwidth]{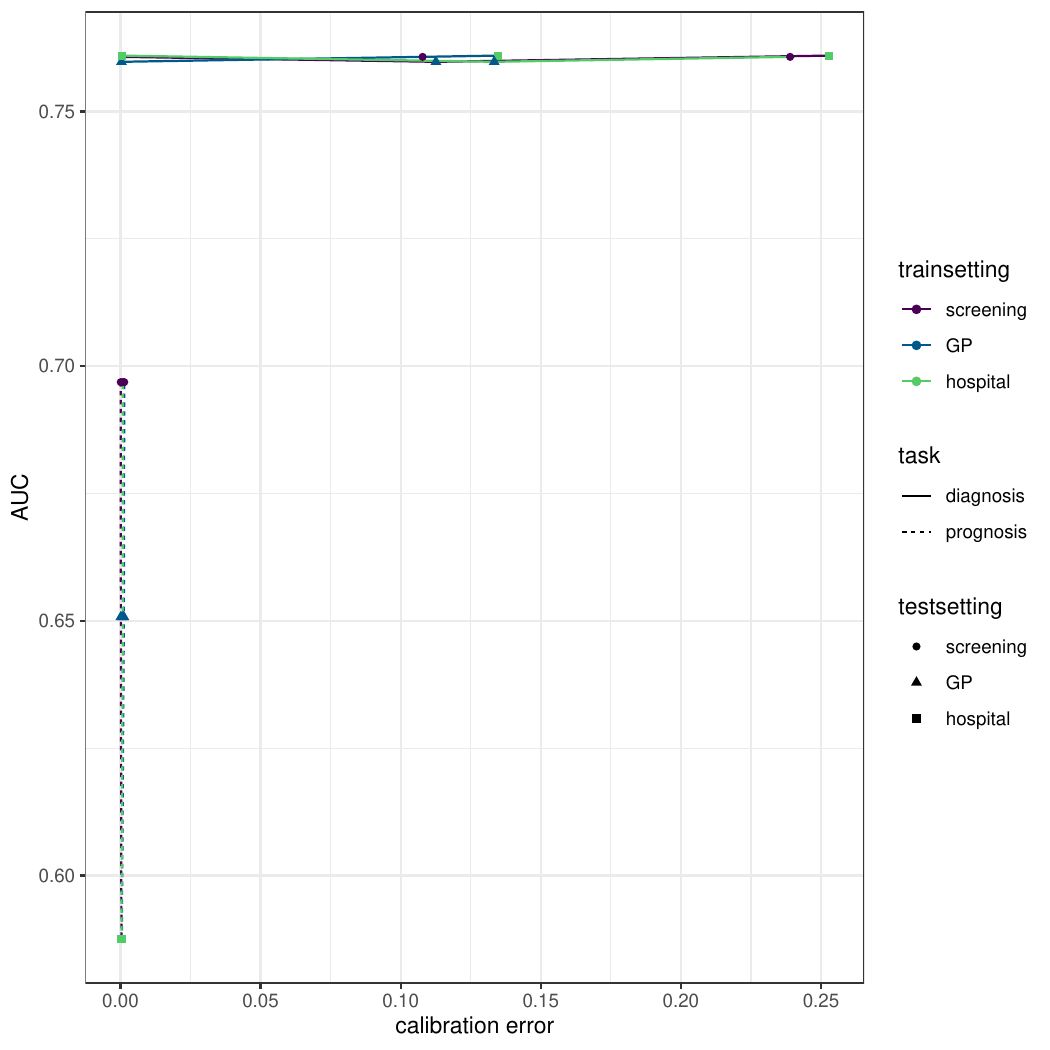}
    \caption{Combined results of the simulation experiment. Each model is connected by a line.}
    \label{fig:combined}
\end{figure}

\subsection{Empirical Study}

As an empirical evaluation we re-used data from a published systematic review on prediction models in cardiovascular disease which included 2030 external validations of 1382 predictions models \cite{wesslerExternalValidationsCardiovascular2021} and whose data is publicly available at https://www.pacecpmregistry.org.
The review investigated changes in model performance when comparing the original publication with later external validation studies.
The authors classified the prediction models as either `diagnostic' or `prognostic' (indicated by a follow-up time of less than 3 months, 3--6 months or more than 6 months).
Selecting only prediction models with one or more validations and information on AUC in both the original study and validation study, and with information on model type (diagnostic versus prognostic), 1170 validation studies remained of 342 prediction models, 16 of which were validation studies of 11 diagnostic models.
Comparing the AUC in the original study ($\text{AUC}_0$) with external validation studies ($\text{AUC}_1$), we calculated the relative difference in AUC as suggested by the authors:

\begin{equation*}
    \delta:=\frac{(\text{AUC}_1 - 0.5) - (\text{AUC}_0 - 0.5)}{\text{AUC}_0 - 0.5}.
\end{equation*}

Our framework predicts that for diagnostic models that predict in the anti-causal direction, the AUC remains the same so $\text{AUC}_0=\text{AUC}_1$, thus $\text{VAR}(\delta)=0$, but not for prognosis models that predict in the causal direction.
The studies in this systematic review are likely not perfectly \emph{causal} or \emph{anti-causal}, and because of sampling variance, variation in AUC will occur.
Still we expect the variance of $\delta$ to be higher for prognosis models than for diagnosis models.
In these data this was indeed the case with $\text{VAR}(\delta_{\text{prognostic}}) \approx 8.2 * \text{VAR}(\delta_{\text{diagnostic}}) = 0.019$, 95\% confidence of ratio: 3.41 - 15.10, p-value for F-test $<0.001$.
Unfortunately the review provided no quantitative measures of calibration so a similar comparison of the variance of changes in calibration could not be made.

The code needed to reproduce the simulation experiment and empirical evaluation are available at Zenodo.

\section{Related work}

The notion that calibration is stable under shifts in the distribution of the \emph{cause} variables has long been appreciated (e.g. \cite{schoelkopfCausalAnticausalLearning2012, piccininniDirectedAcyclicGraphs2020, waldCalibrationOutofDomainGeneralization2021}).
Much of the theory in this paper is inspired by Sch{\"o}lkopf's work on causal and anti-causal learning \cite{schoelkopfCausalAnticausalLearning2012}.
This work connects the general framework to the medical setting in two ways: we define a familiar term `case-mix' in a formal causal language, and then derive how two canonical metrics of predictive performance (discrimination and calibration) respond differently to changes in case-mix for \emph{causal} prediction models (prognosis) and \emph{anti-causal} models (diagnosis).

Prior work noted that prediction models that are calibrated in multiple environments are provably free from anti-causal predictors \cite{waldCalibrationOutofDomainGeneralization2021}. Our focus is in the reverse direction: when to expect stable calibration across environments.
Jaladoust and colleagues derived general bounds for functionals of the target distribution in a new environment \cite{jalaldoustPartialTransportabilityDomain2024}. 
Our work describes when certain specific functionals from the distributions (discrimination and calibration) are stable across settings, tailored to typical needs in the (medical) prediction model setting.
Other prior work requires detailed assumptions on the causal relationships between variables \cite{subbaswamyUnifyingCausalFramework2022}, or access to data from multiple environments.
Subbaswamy consideres loss functions of the form $l(\hat{y},y)$ which implicitly depend on the joint distribution of $X,Y$ trough the expectation $L=E_{X,Y}[l(\hat{y},y)]$.
Our work is focused on metrics that are direct functionals of the conditional distribution $Y|X$ and $X|Y$.
Also, Subbaswamy focusses on min-max optimality.

Our framework also provides a new perspective on the results of the study by Fehr at al \cite{fehrAssessingTransportabilityClinical2023}.
They experimented with prediction models that contained either causal factors of the outcome (related to our \textit{prognosis} models), anti-causal factors (related to our \textit{diagnosis} models), or a combination of both.
The performance of different prediction models was evaluated under different shifts in variables that were at the same time a direct cause of the outcome and a cause of other variables.
Fehr et al found that for models predicting only with cause variables, the calibration is stable under interventions on only cause variables, as directly explained by our main result.
When predicting with anti-causal factors, they observed that under interventions on the cause variables, the calibration degrades for models that are well calibrated on the training data.
This setting is the closest to our \textit{diagnostic} setting, though technically it is a mix of the anti-causal DAG \ref{fig:anticausal} and the confounded DAG \ref{fig:fork}.

\section{Discussion}

We present a novel causal framework for understanding changes in prediction model performance under shifts in case-mix, by defining a shift in case-mix as a change in the marginal distribution of the cause variable.
This leads to a new understanding of why in certain situations the discrimination of a model may be relatively stable when evaluated in a different setting, but not the calibration, and vice-versa.

Limitations are that the definition of a shift in case-mix is an abstraction and pure interventions on only either the features or the outcome may be unrealistic in practice.
Many diagnostic prediction models may contain features that have a causal path to the diagnosis (e.g. age), or `risk factors' for the disease that are not caused by the presence or absence of the disease.
Also, the `disease' itself may be an abstraction, and the diagnosis used in medical practice may be a combination of effects of an underlying biological process.
Systematic reviews of diagnostic models indeed show variation in sensitivity and specificity with variation in disease prevalence, a phenomenon also referred to as the \textit{spectrum-effect} \cite{leeflangVariationTestsSensitivity2013}.
Still, when compared with prognostic models, diagnostic models had lower variability in discrimination in our empirical study.
The current empirical evaluation was limited, and classifying diagnostic models as anti-causal and prognostic models as causal may be too crude. Also, no quantitative data on calibration were available to test whether calibration was more stable for prognostic models.
Future empirical studies of externally evaluated prediction models will shed more light on how this theory pans out in practice.
Mis-calibration may occur when variables not included in the model are also shifted between environments.

What are the implications of the causal case-mix framework for different stakeholders?
A main use of this new causal case-mix framework is to provide an explanation of (lack of) expected and observed differences in prediction model performance across environments.
For prediction model developers, this framework provides a new way to think about the features included in a prediction model.
In some settings such as triaging patients in emergency room for early medical evaluation, the utility of a prediction model depends mostly on its discrimination. In other settings such as cardiovascular risk management, a prediction model's utility depends on its calibration. Depending on this utility function, prediction model developers may opt to include mostly causal or anti-causal features in a prediction model, if dependable performance across environments is desired.
The framework adds another perspective on the discussion on when and where to re-calibrate a prediction model \cite{swaminathanReflexiveRecalibrationCausal2025,lekadirFUTUREAIInternationalConsensus2025}.
When calibration is important and the model does include anti-causal features, it is likely necessary to always recalibrate the model when taking it to a new environment, in line with many recommendations \cite{youssefExternalValidationAI2023}.
However, when discrimination changes for a prediction model in the causal direction, this may not warrent a re-fitting of the model as this change is to be expected under a shift in case-mix.
A decrease (or increase) of discrimination may be indicative of a shift in the \emph{data}, meaning more (or less) patients with probabilities closer to 50\%, than of a bad model.

For researchers that evaluate prediction models and policy makers, it was long known that no models are robust to arbitrary changes in distribution. This framework implies that a subset of models should have stable calibration or discrimination. For example, observing a stable discrimination of a diagnostic model with anti-causal features in several different environments may provide confidence that the model is indeed robust to environmental changes. At the same time, this model's calibration should \emph{not} be stable across evaluations. Whereas normally a stable calibration would be seen as a re-assuring sign, having both calibration and discrimination stable across environments is a sign that the environments are not meaningfully different at all. A stable discrimination paired with unstable calibration (or the other way around) is a stronger sign of robustness to changes in environment than when both are stable, as this would only occur when the environment are too similar.

\section*{Acknowledgments}

The author kindly acknowledges Anne de Hond, Oisin Ryan and Valentijn de Jong for fruitful discussions on this framework.                                                                                                                                                                          

\printbibliography

\section*{Appendix}

\subsection*{Definitions}\label{app:defs}
\begin{definition}[calibration]\label{def:calibration}
  Let $P(X,Y)$ be a joint distribution over feature $X$ and binary outcome $Y$, and $f: \mathcal{X} \to [0,1]$ a deterministic prediction model. \emph{$f$ is perfectly calibrated with respect to $P(X,Y)$} if, for all  $\alpha \in [0,1]$ in the range of $f$, $\EX_{X,Y \sim P(X,Y)}[Y|f(X)=\alpha]=\alpha$.
\end{definition}

\begin{definition}[case mix]\label{def:casemix}
	Let $Z,X,Y$ be random variables and $E$ an environment variable. Assume one of the three following causal directed acyclic graphs labeled \emph{causal}, \emph{anti-causal} and \emph{fork} (shown also in Figure \ref{fig:dags}):
	\begin{enumerate}
		\item causal: $E \to X \to Y$
		\item anti-causal: $E \to Y \to X$
		\item fork: $E \to Z \to X; Z \to Y$
	\end{enumerate}
	Let $P_E(.)$ denote the distribution of variable $.$ in environment $E$.
	A \emph{shift in case-mix} across environments $e,e' \in \mathcal{E}$ is a shift in the distribution of the direct child of $E$ in the DAG, meaning a shift in:
	\begin{enumerate}
		\item $P_E(X)$ when DAG = causal
		\item $P_E(Y)$ when DAG = anti-causal
		\item $P_E(Z)$ when DAG = fork
	\end{enumerate}
\end{definition}

\begin{remark}[conditional independencies]\label{prop:indeps}
	The causal DAGs enumerated in Definition \ref{def:casemix} imply the following conditional independencies regarding random variables $X,Y$:

\begin{table}[htp]
    \centering
\begin{tabular}{|l|l|l|l|}
\hline
                         & causal              & anti-causal         & fork                \\ \hline
$P_E(X)$ vs $P(X)$       & $=$                 & $\color{red}{\neq}$ & $\color{red}{\neq}$ \\ \hline
$P_E(Y)$ vs $P(Y)$       & $\color{red}{\neq}$ & $=$                 & $\color{red}{\neq}$ \\ \hline
$P_E(Y|X)$ vs $P(Y|X)$   & $=$                 & $\color{red}{\neq}$ & $\color{red}{\neq}$ \\ \hline
$P_E(X|Y)$ vs $P(X|Y)$   & $\color{red}{\neq}$ & $=$                 & $\color{red}{\neq}$ \\ \hline
\end{tabular}
\end{table}

\end{remark}

\subsection*{Theorems}
\label{app:theorems}

We now describe our main result.

\begin{theorem}[perfectly calibrated models remain perfectly calibrated under marginal shifts in $X$]
	\label{th:calib}
    Given binary $Y$, prediction model $f: \mathcal{X} \to [0,1]$ and environment $E \in \{\text{train},\text{test}\}$.
    Assume $X$ takes on values from a measureable space $\mathcal{X}$ with measures $\phi_{train}(x), \phi_{test}(x)$ on the training and testing environment, and assume $\phi_{test}(x) << \phi_{train}(x)$ (the support of the test distribution is contained in the support of the training distribution).
	Assume $P_{E}(Y,X) = P_{E}(X)P(Y|X)$ (shift in $P(X)$ but not $P(Y|X)$ between environments).
    Define the \emph{miscalibration} of $f$ under $P_{E}$ for value $X=x$ as:

	\begin{equation}
		\label{eq:xi}
        \xi(x) := |f(x) - P(Y=1|X=x)|
	\end{equation}

	Then the integrated calibration index (ICI) \cite{austinIntegratedCalibrationIndex2019} on distribution $P_E$ is:
	\begin{equation}
		\label{eq:ece_calib}
		\text{ICI}_{E} = \EX_{X \sim P_{E}(X)} \xi(x)
	\end{equation}

    Theorem statement:
    \emph{a model that is perfectly calibrated on the training distribution (i.e. $\xi(x) =  0 \iff \phi_{train}(x) > 0$) remains perfectly calibrated in the test distribution}:
	\begin{equation}
		\label{eq:ece_perfect}
        \text{ICI}_{\text{train}} = \text{ICI}_{\text{test}} = 0
	\end{equation}

\end{theorem}

\begin{proof}[Proof of theorem \ref{th:calib}]
  By assumption we have $\phi_{train}(x) > 0 \implies \xi(x) = 0$, because $\phi_{train}(x) > 0 \implies \phi_{test}(x) > 0$ we also have that $\phi_{test}(x) >0 \implies \xi(x) = 0$ ($f$ is calibrated for all values of $x$ in the test distribution).
  Denote $\text{supp}_{test}(X)$ the subset of $\mathcal{X}$ where $\phi_{test}(x) > 0$. By definition of ICI we have that
    \begin{align}
        \text{ICI}_{\text{test}} &= \EX_{X \sim P_{\text{test}}(x)} \xi(x) \\
                                 &= \int_{\text{supp}_{test}(X)} \xi(x) d\phi_{\text{test}}(x) \\
                                 &= \int_{\text{supp}_{test}(X)} 0 d\phi_{\text{test}}(x) \\
                                 &= 0 \int_{\text{supp}_{test}(X)} d\phi_{\text{test}}(x) \\
                                 &= 0 * 1 \\
                                 &= 0
    \end{align}
\end{proof}

\begin{theorem}[discrimination is constant under marginal shifts in $Y$]
	\label{th:discr}
    Given binary outcome $Y$, prediction model $f: \mathcal{X} \to [0,1]$ and environment $E \in \{\text{train},\text{test}\}$.
	Assume $P_{E}(Y,X) = P_{E}(Y)P(X|Y)$ (marginal shift of $Y$ but not $X|Y$).
	Furthermore assume $0<P_{E}(Y=1)<1$ (marginal distribution of $Y$ in both distributions is non-deterministic).
	Then for all thresholds $0 \leq \tau \leq 1$:
	\begin{align}
    \text{sens}_{\text{test}}(\tau) &= \text{sens}_{\text{train}}(\tau) \\
    \text{spec}_{\text{test}}(\tau) &= \text{spec}_{\text{train}}(\tau)
	\end{align}

	And also
	\begin{equation}
        AUC_{\text{train}} = AUC_{\text{test}}
	\end{equation}
	
\end{theorem}

\begin{proof}[proof of theorem \ref{th:discr}]
Theorem \ref{th:discr} follows directly from the fact that sensitivity: $P(f(X)>\tau|Y=1)$, specificity: $P(f(X) \le \tau | Y=0)$ and $P_{\text{train}}(f(X)|Y=y) = P_{\text{test}}(f(X)|Y=y)$.

\end{proof}

\subsection{Examples}
\label{app:examples}

Clinical examples of when the definition of a shift in case-mix as in Definition \ref{def:casemix} may apply across different environments.

\subsubsection{Examples in the causal direction}

Prediction of the occurrence of a cardiovascular event in the coming 10 years based on age and the presence of diabetes at baseline.
\begin{enumerate}
    \item train environment: general practitioner
    \item test environment: a diabetes out-patient clinic
\end{enumerate}

\subsubsection{Examples in the anti-direction}

Example 1: prediction of the presence of a stroke based on computed tomography imaging of the brain:
\begin{enumerate}
    \item train environment: secondary care hospital
    \item test environment: stroke center where patients are referred when they have stroke symptoms
\end{enumerate}

Example 2: Diagnosing sexually transmittable disease (see Figure \ref{fig:exstd}).

 \begin{figure}[ht!]
	\begin{subfigure}[b]{0.5\textwidth}
     \centering
	     \begin{tikzpicture}
             \node[text centered] (x) {$X$};
             \node[right of = x, node distance=1.5cm, text centered] (y) {$Y$};
             \node[right of = y, node distance=1.5cm, text centered] (y0) {$Y_0$};
             \node[draw, rectangle, below of = y, node distance=1.5cm, text centered] (e) {$E$};
             
             \draw[->] 
                 (e) edge (y0)
                 (y0) edge (y)
                 (y) edge (x);
	     \end{tikzpicture}
	\caption{example with all model variables}
    \label{dag:std1}
     \end{subfigure}
\hfill
	\begin{subfigure}[b]{0.5\textwidth}
     \centering
	     \begin{tikzpicture}
             \node[text centered] (x) {$X$};
             \node[right of = x, node distance=1.5cm, text centered] (y) {$Y$};
             \node[draw, rectangle, below of = y, node distance=1.5cm, text centered] (e) {$E$};
             
             \draw[->]
                 (y) edge  (x)
                 (e) edge  (y);
	     \end{tikzpicture}
	\caption{resulting DAG when marginalizing over $Y_0$}
    \label{dag:std2}
     \end{subfigure}
     \caption{Example setting of diagnosing a sexually transmittable disease (STD, $=Y$) with a blood test ($=X$) in either general public setting (\ref{dag:std1}) or in a HIV-positive clinic (\ref{dag:std2}). Patients with previous STDs such as HIV ($Y_0$) have a higher risk of future STDs, summarized with the arrow from $Y_0$ to $Y$. $Y_0=1$ is a selection criterion for the HIV-clinic, meaning that only patients with a prior STD get seen at the HIV-clinic.
 Treating $Y_0$ as not observed (thus marginalizing it out) results in the DAG in \ref{dag:std2}} 
     \label{fig:exstd}
 \end{figure}
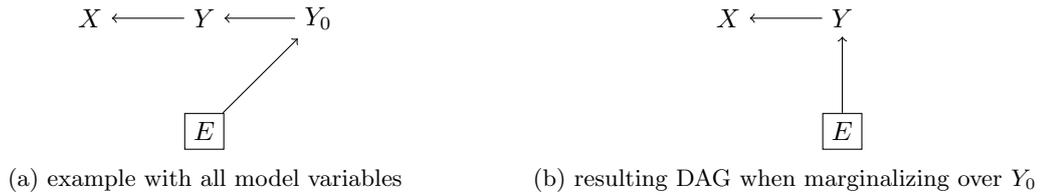

\subsection{Additional figures of simulation study}
\label{app:sims}

\begin{figure}[htpb]
    \centering
    \includegraphics[width=0.8\textwidth]{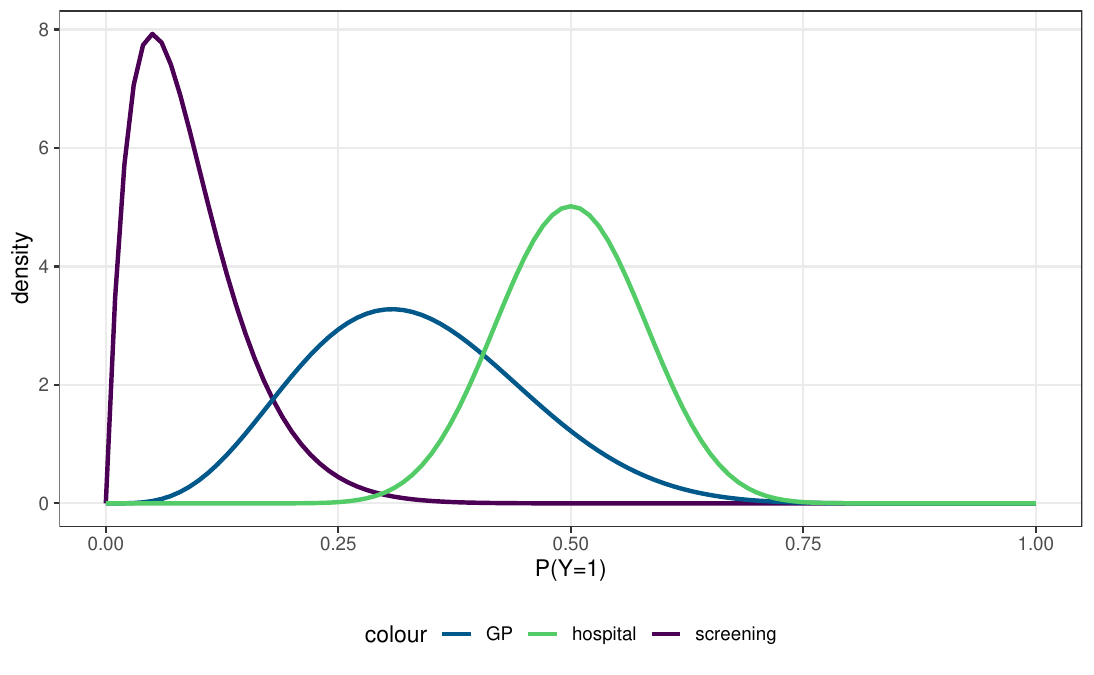}
    \caption{Marginal distribution $P(Y=1)$ in different environments, given by beta-distributions with parameters listed in Table \ref{tab:simprms}.}
    \label{fig:fig-py-prognosis-pdf}
\end{figure}

\begin{figure}[htpb]
    \centering
    \includegraphics[width=0.8\textwidth]{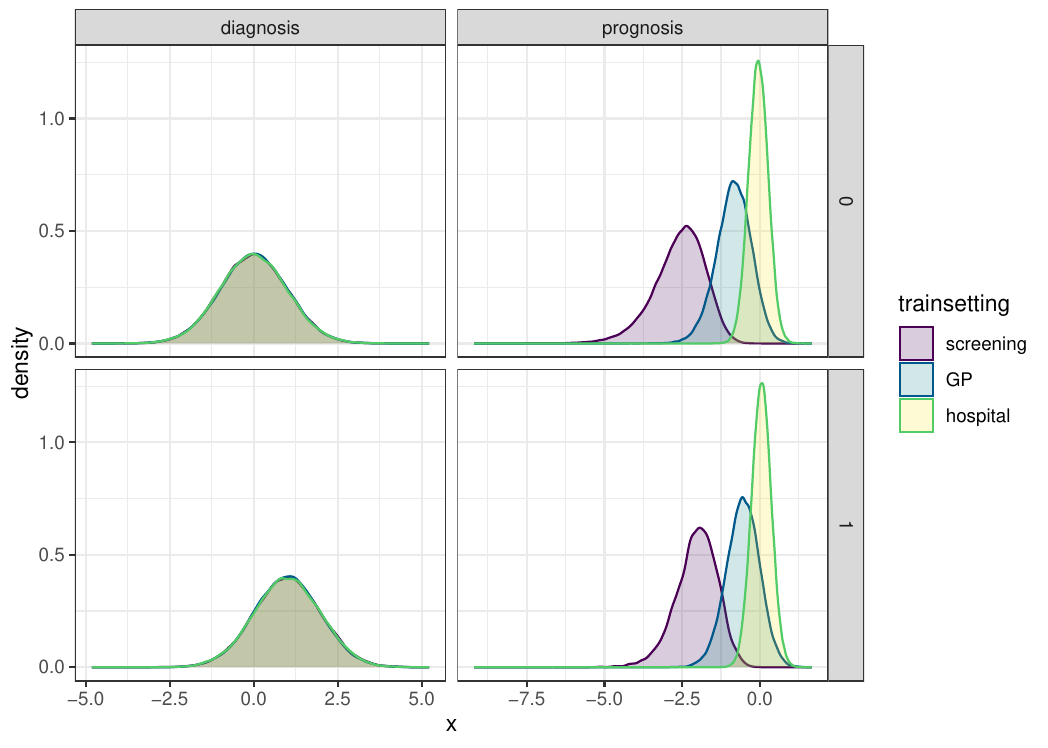}
    \caption{Conditional density of $P(X|Y=y)$ in the diagnosis or the prognosis simulation setting}
    \label{fig:fig-x-given-y-pdf}
\end{figure}

\begin{figure}[htpb]
    \centering
    \includegraphics[width=0.8\textwidth]{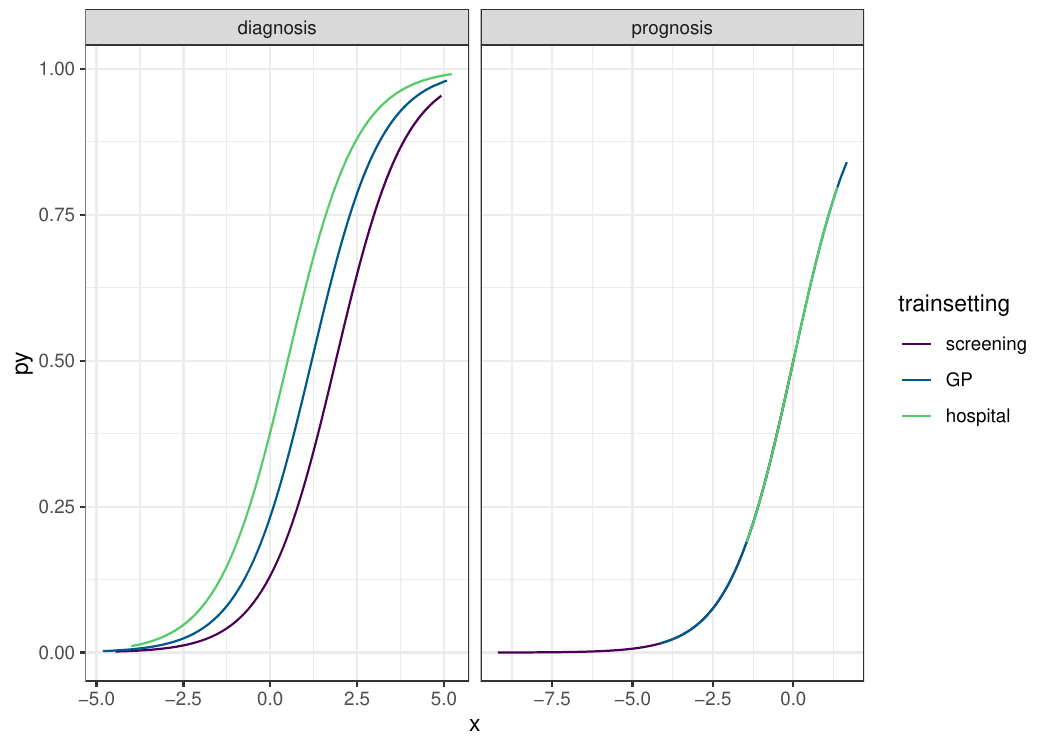}
    \caption{Conditional distribution of $P(Y=1|X=x)$ in the diagnosis or the prognosis simulation setting.}
    \label{fig:fig-y-given-x-pdf}
\end{figure}

\end{document}